\documentclass[final]{amsart}
\usepackage{amssymb,amsfonts}
\usepackage{amsmath}
\usepackage{enumitem}
\usepackage[colorinlistoftodos]{todonotes}
\usepackage{float}
\usepackage{xcolor}
\usepackage{mathtools}
\usepackage{nicefrac}
\usepackage{hyperref}
\usepackage[foot]{amsaddr}

\definecolor{minttext}{RGB}{0,204,171}

\newenvironment{enumeratei}{\begin{enumerate}[label=\textup{(\roman*)}, noitemsep, topsep=1.5mm, labelindent=.8em, leftmargin=*, widest=.]}{\end{enumerate}}
\newenvironment{enumeratenum}{\begin{enumerate}[label=\textup{(\arabic*)}, noitemsep, topsep=1.5mm, labelindent=.8em, leftmargin=*, widest=.]}{\end{enumerate}}

\hypersetup{
    colorlinks=true,
    citecolor=minttext,
    linkcolor=minttext,
}

\usepackage[ruled,titlenumbered,linesnumbered,noend,norelsize]{algorithm2e}
\newenvironment{algorithm-hbox}{\hbadness=10000\begin{algorithm}}{\end{algorithm}}

\newtheorem{theorem}{Theorem}[section]
\newtheorem{corollary}[theorem]{Corollary}
\newtheorem{definition}[theorem]{Definition}
\newtheorem{lemma}[theorem]{Lemma}
\newtheorem{remark}[theorem]{Remark}
\newtheorem{notation}[theorem]{Notation}

\renewcommand{\leq}{\leqslant}
\renewcommand{\geq}{\geqslant}

\renewcommand{\H}{\texttt{hash}}

\newcommand{\pa}{\mathrm{R}_{U_C}}

\renewcommand{\P}{\mathcal P}
\newcommand{\U}{\mathcal U}

\newcommand{\R}{\mathbb R}
\newcommand{\TC}{\mathrm{TC}}

\newcommand{\val}{\texttt{validated}}
\renewcommand{\L}{\mathrm{L}}

\newcommand{\SupMaj}{\rm SupMaj}
\newcommand{\Exist}{\rm Exist}

\newcommand{\fai}{\rm{FAI}}
\newcommand{\pri}{\rm{prime}}
\newcommand{\supp}{\rm supp}
\renewcommand{\frac}{\nicefrac}

\newcommand{\down}[1]{\downarrow_1\mathclose(#1)}

\title[Aleph Consensus Protocol]{ Aleph: A Leaderless, Asynchronous, Byzantine Fault Tolerant Consensus Protocol}

  \author{Adam G\c{a}gol$^{\ast\dagger}$}
  \author{Michal \'Swi\c{e}tek$^{\ast\dagger}$}
  \date{\today}
  \address[$\ast$]{Aleph Zero Foundation, alephzero.org}
  \address[$\dagger$]{Jagiellonian University, Krakow, Poland}
  
\begin{document}

\begin{abstract}
    In this paper we propose Aleph, a leaderless, fully asynchronous, Byzantine fault tolerant consensus protocol for ordering messages exchanged among processes.
    It is based on a distributed construction of a partially ordered set and the algorithm for reaching a consensus on its extension to a total order.
    To achieve the consensus, the processes perform computations based only on a local copy of the data structure, however, they are bound to end with the same results.
    Our algorithm uses a dual-threshold coin-tossing scheme as a randomization strategy and establishes the agreement in an expected constant number of rounds.
    In addition, we introduce a fast way of validating messages that can occur prior to determining the total ordering.
    
    {\color{red} This version of the protocol is deprecated. For current version see \cite{glss}}.
\end{abstract}

\maketitle

\section{Introduction}\label{sec-introduction}

State machine replication is a paradigm for implementing fault-tolerant services using distributed servers communicating over an unreliable network.
At the heart of such systems lie consensus algorithms responsible for establishing an agreement among replicas on a value or an action.
In a real life scenario a certain number of replicas can be faulty and the communication between them can suffer delays. 
Therefore, one of the crucial properties of consensus algorithm is their resilience to node failure and network delay. 
The strongest guarantees of such a resilience are, respectively, Byzantine fault tolerance and full asynchronicity.
The former assumes that faulty processes can deviate from the protocol in a completely arbitrary way, from simple crashes, to malicious behavior aimed at disturbing the consensus, fully coordinated between all faulty processes.
Full asynchronicity assumes that delays on messages sent between users can be arbitrarily long, but finite (every message will eventually be delivered).
For an overview of Byzantine consensus algorithms over asynchronous networks see \cite{CVNV}.

One of the first notable examples of consensus algorithms is the original Paxos \cite{Lamport} algorithm and its later incarnation Raft \cite{OO}.
The main downside of their approach is the requirement to choose a leader who is responsible for coordination of establishing consensus.
The problem of such an approach is evident, since at any point in time the designated leader is a single point of failure of the whole protocol.

A different approach, Nakamoto consensus \cite{Sat}, based on blockchain, resolves the problem of choosing a leader in a sense that the process producing a new block is impossible to predict beforehand.
As such, the protocol is immune to a single point of failure.
However, systems based on Nakamoto consensus suffer from serious scalability issues and their intrinsic tradeoff between security and speed \cite{Pass}.
Moreover, due to the longest chain rule, processes can never be sure whether the consensus has been reached, it is only guaranteed with high probability.

Another crucial property of a consensus algorithm is liveness, or a guarantee that agreement is eventually reached between all non-faulty processes.
The celebrated FLP theorem \cite{FLP} states that no deterministic and asynchronous algorithm can achieve liveness in the presence of at least one faulty process.
A standard technique of circumventing FLP theorem is to give up determinism by introducing randomization in a form of coin tossing.
Such randomized Byzantine consensus algorithms can be divided into two main classes based on their randomization strategy.
In the first class, all processes have access to one shared coin; in the second class, each process uses its own local coin.
The main difference between these two classes lies in the expected number of coin tossing rounds required for an algorithm to terminate.
Shared coin protocols can terminate in expected constant number of rounds, while local coin protocols typically require an expected exponential number of rounds.
On the other hand, the shared coin strategy typically requires a trusted dealer to distribute coin shares among the replicas.

Hashgraph \cite{Baird} is an example of a consensus algorithm that satisfies all the aforementioned properties, i.e. it is randomized, leaderless, asynchronous, and Byzantine fault tolerant.
However, it uses a local coin randomization strategy and has an expected exponential number of rounds needed to terminate\footnote{This can be seen by computing the exponentially small probability of the Hashgraph's coin round settling the consensus in case of $\frac{1}{3}N - 1$ Byzantine processes.}.

In this paper, we introduce Aleph, a new distributed consensus algorithm that suffers none of the above shortcomings.
Our protocol is leaderless, asynchronous, Byzantine fault tolerant and guarantees liveness with an expected constant number of rounds needed to terminate.  
We use a technique inspired by the idea of a dual-threshold coin-tossing scheme from \cite{Cachin}, modified in a way that removes the reliance on a trusted dealer in the initialization phase, making the Aleph protocol fully leaderless.
In our communication protocol, processes exchange information by creating and sharing with each other units which can contain arbitrary data.
At the moment of their creation, newer units are connected to the older ones and a structure emerges in a similar fashion to events partially ordered by the happened-before relation introduced by Lamport \cite{Lamport-clocks}.
The total ordering of units is computed locally, based only on the partial order structure, and is guaranteed to be the same for all processes.
This approach, similar to Hashgraph's virtual voting \cite{Baird}, allows for constant average communication complexity when high message throughput is present. 

In addition, we introduce a fast way of validating certain types of messages.
It can be performed prior to establishing the total ordering and may be used in applications such as token transfer systems.

The rest of the paper is organized as follows.
In Section \ref{sec-preliminaries}, we introduce necessary cryptographic tools and formally define the network model.
The basic communication protocol with its key properties is presented in Subsections \ref{subsec-definitions}-\ref{subsec-structural}.
Subsections \ref{subsec-leaderless-coin} and \ref{subsec-timing} are devoted to the core ordering protocol and proving its properties.
Finally, in Subsection \ref{subsec-validation}, we introduce the concept of fast unit validation.

\section{Preliminaries}\label{sec-preliminaries}

\subsection{Cryptographic Primitives}\label{subsec-primitives}

\subsubsection{Digital Signatures}

A digital signature scheme can be used for message authentication.  
The Aleph protocol works on top of any digital signature scheme with the following properties:
\begin{enumeratei}
\item[(1)] provides a verifier a proof that the signed message was created by a known sender (authentication),
\item[(2)] does not allow a signer of a message from denying the creation of the message (non-repudiation), and
\item[(3)] does not allow the message to be altered in any way (integrity).
\end{enumeratei}

\subsubsection{Threshold Coin-Tossing Scheme}\label{coin-description}

The concept of $(n,t,k)$ dual-threshold coin-tossing scheme was first introduced by Cachin, Kursawe and Shoup \cite{Cachin}. 
It allows $n$ parties to share a coin toss in a scenario where up to $t$ of them can be corrupted. 
The parties generate shares of the coin based on a nonce $C$, and $k$ such shares are required to construct the value of the coin for that particular nonce. 
In our scenario, we use $t=\frac{1}{3}N-1$ and $k=\frac{2}{3}N+1$.

The original threshold coin requires a trusted dealer to distribute keys to all participants.
To circumvent this requirement, we use $n$ different threshold coins, one for each process participating in the consensus.
We introduce a notation $TC^i(C)$ to denote the value of the threshold coin dealt by $i$-th process for a nonce $C$.
Moreover, we use $TC^i_j(C)$ to denote the $j$-th process's share of the threshold coin dealt by the $i$-th process for a nonce $C$.

\subsubsection{Common Random Permutation}

A common permutation scheme is a simple procedure producing a sequence of permutations of the processes based on the set of their public keys. 
The algorithm takes a nonce (natural number) and the set of all public keys as input and outputs a pseudo-random permutation of the processes, which cannot be controlled without controlling all the public keys.

For a nonce $k$, the permutation $\sigma_k$ is constructed by first computing $X$ as the XOR of all public keys.
Then, each public key $VK_i$ is iteratively hashed $k$ times to compute the value $X_i = \mathrm{hash}^k(VK_i)$ for $i\in \{1,\dots, n\}$.
Finally, the value $X_i$ is XORed with $X$ to produce $Y_i = X_i \texttt{XOR} X$.
The lexicographic ordering of all $Y_i$'s is the desired permutation $\sigma_k$.
Note that any party that knows all the public keys is able to locally compute such permutation for every nonce $k$.

\subsection{Model and problem statement}\label{subsec-model}

The Aleph protocol operates on a network of $N$ processes and tolerates up to $\frac{1}{3}N-1$ faulty processes.
Faulty processes can deviate from the protocol in an arbitrary way, including both crashes and Byzantine faults.
In other words, we assume $\frac{1}{3}N-1$ processes are controlled by a malicious adversary.
All remaining $\frac{2}{3}N+1$ processes are \emph{correct} and fully follow the protocol.

We assume each pair of processes is connected by a reliable authenticated channel, but the delivery schedule is under full control of the adversary, i.e., we consider a fully asynchronous scenario. 
Moreover, the adversary is able to listen to all communication, but cannot modify its content\footnote{Such channels can be easily emulated on top of unreliable channels via digital signature system and resending transmission or by applying error-correcting codes.}. 

The Aleph protocol aims to establish an atomic broadcast (or total order broadcast) between all processes in the network. 

\begin{definition}\label{def-atomic-broadcast}
    An \emph{atomic broadcast} is a protocol of message broadcasts between processes satisfying the following properties in the presence of a malicious adversary: 
    \begin{enumeratei}
        \item {\bf Correctness} - all correct processes need to agree on validity and ordering of all messages,
        \item {\bf Integrity} - once a message $T$ is agreed to be placed in the ordering, the ordering of $T$ and all messages before it is final, i.e., no more messages can be decided to be earlier than $T$, and
        \item {\bf Liveness} - every message proposed by a correct process will eventually be agreed upon by all correct processes.

    \end{enumeratei}
\end{definition}

\section{Aleph Protocol}\label{sec-aleph-protocol}

\subsection{Basic Definitions}\label{subsec-definitions}

The basic building block of the protocol is a \emph{unit}, a data structure created by a single process and propagated throughout the network.
Units are used as containers for messages, and the protocol uses them to guarantee the consistency of the system.

\begin{definition}
    An object $U$ of a class {\normalfont \texttt{unit}} has the following attributes: 
    \begin{enumeratei}
        \item[] {\normalfont \texttt{U.parents}} - hashes of two parent units, 
        \item[] {\normalfont \texttt{U.message}} - a message contained in the unit,
        \item[] {\normalfont \texttt{U.creator}} - an ID of the process that created $U$,
        \item[] {\normalfont \texttt{U.signature}} - a digital signature of the unit creator,
        \item[] {\normalfont \texttt{U.coinshares}} - a set of threshold coin shares attached to the unit.
    \end{enumeratei}
\end{definition}

There are two special types of units that can contain some additional attributes: the genesis unit and dealing units.
The \emph{genesis unit} is the first created unit.
For each process, a \emph{dealing unit} is the first unit created by that process after the genesis unit.

We also note that all attributes are set by the unit creator and never change.
In addition, by the hash of a unit $U$, denoted as \texttt{hash(U)}, we mean the value of a fixed cryptographic hash function (for example, \texttt{sha256d}) on the bit string representation of the unit object.  We also observe that as time progresses, the set of all units grows, since the algorithm only creates units and never destroys them.

\begin{notation}\label{notation-poset}
    Let $\P^t$ denote the set of all units at a time $t$.
    If the time is not relevant or it is known from the context, we will omit the superscript $t$ and write just $\P$. 
\end{notation}

Since every unit contains hashes of two other units, a structure is imposed on the whole set of units.
This structure is called a \emph{partial order} and is defined in the following way.

\begin{definition}
    Let $U$ and $V$ be units in the set $\P$.  
    We say that $V$ is {\em above} $U$ and denote it as $V \geq U$ if $U = V$ 
    or there exists a sequence of units, $U_0, U_1, \dots, U_n \in \P$, such that $U_0=U,\, U_n=V$, and ${\normalfont\H(U_{i-1}) \in U_{i}\texttt{.parents}}$ for all $1\leq i\leq n$.
\end{definition}

The set $\P$ together with the partial order $\leq$ is called a \emph{partially ordered set}, or a \emph{poset} for short.
We will abuse the notation and call the set $\P$ a poset if it does not lead to confusion.

A set $\mathcal{U}$ of units forms a \emph{chain} in a poset $\P$ if for all 
$U, V \in \mathcal{U}$, $U \leq V$ or $V \leq U$, i.e. any two units in the set are comparable by the relation $\leq$.

Due to the latency of information exchange, processes in the network will not have instantaneous knowledge about every unit.
Therefore, at a time $t$, each process will be aware of only a part of the whole poset $\P^t$ (a great majority of it, but with exclusion of some of the most recent units)
and the following notion arises naturally.

\begin{notation}
    For a process $A$ at a time $t$, we define $A$'s \emph{local view} of the poset $\P^t$ as the part of $\P^t$ that $A$ knows and denote it as $\P^t_A$.
    We use the term \emph{global view} for the whole poset $\P^t$.
\end{notation}

For a set of units $\mathcal{U}$, we define a \emph{supporting set} for $\mathcal{U}$, denoted by $\supp(\mathcal{U})$, to be a set of all processes which created at least one unit in $\mathcal{U}$. 
Since we will often be interested only in the size of the support for the given set, we introduce the following short notation:
\[
    \#_s\, \mathcal{U} = \#\supp\big(\mathcal{U}\big).
\]

\begin{definition}\label{def-high-above}
    Let $U, V \in\P$ be units, $U \leq V$, and $\mathcal{U} = \{W \mid  U \leq W \leq V\}$.
    We say that the unit $V$ is \emph{high above} the unit $U$ if $\#_s\,\mathcal{U} \geq \frac{2}{3} N$.
    We denote this as $V \gg  U$.
    In addition, we will also say that $U$ is {\em high below} $V$.
\end{definition}

The next few definitions introduce additional structure on the poset $\P$ and is used to track information as it is spread throughout the network, as explained further in Lemmas \ref{lemma-22} and \ref{lemma-two-levels-above}.

\begin{definition}
For a unit $U$, the \emph{level} of $U$ is recursively defined as:

\begin{equation*}
    \L(U) =
    \begin{cases*}
      0 & if U \text{ is the genesis unit, or} \\
      m & if $\#_s\Big\{V\ll U: \L(V) =  m\Big\} < \frac{2}{3}N$ \\
      m+1 & otherwise
    \end{cases*}
\end{equation*}
\end{definition}

\noindent where $m = \max\limits_{V<U}\L(V)$.

\begin{definition}
A unit $U$ is a \emph{prime unit} if there does not exist a unit $V<U$ issued by $U$'s creator such that $\L(U) = \L(V)$.
\end{definition}

\begin{definition}
    Let $U$ and $V$ be prime units.
    If $U\ll V$ and $\L(V) = \L(U)+1$, then $U$ is called a \emph{prime ancestor} of $V$.
    Sets of all prime ancestors of a unit $U$ is denoted as $\down{U}$.
\end{definition}

For technical reasons, we will need to sum a series of the following type.

\begin{lemma}\label{lemma-spanish}
    Let $q\in\R$ with $|q|<1$.
    Then
    \[
        \sum_{n \geq 0}nq^n = \tfrac{q}{(1-q)^2}.
    \]
\end{lemma} 

As a notational convention, if $f$ is a function with a domain $X$ and $S\subset X$, we use $f(S)$ to denote the image of $S$ by $f$.

Finally, to make formulas with predicates easier to read, we use the Iverson bracket convention throughout the paper:

\begin{equation*}
    [P] =
    \begin{cases*}
      1 & if predicate $P$ is true\\
      0 & if predicate $P$ is false.
    \end{cases*}
\end{equation*}

\subsection{Communication Protocol}\label{subsec-communication}

The Aleph protocol organizes the exchange of messages by participants (represented as processes) through the use of atomic message containers called units and establishes a consensus on the ordering of all the messages.
The core communication protocol consists of two parts: (1) locally building the poset by adding new units to it (Algorithm \ref{alg_create} \texttt{create\_unit}) and
(2) synchronizing the local views among the processes (Algorithm \ref{alg_sync} $\texttt{sync}$). 
We note that every process in the network starts by running an infinite loop. 
At every iteration, a process runs the $\texttt{create\_unit}$ function, and at every $K$ steps, it runs the $\texttt{sync}$ function, where $K$ is a global parameter.
All the algorithms will be written in the convention that every function will have the identifier of a process as a first parameter.

New units are created by a process $A$ via Algorithm \ref{alg_create}.
Let $V$ be the last unit created by $A$.
First, the process $A$ randomly chooses $U_1$, the first parent of a new unit $U$, to be a maximal element in his local view which is above $V$ (line $1$).
Next, $A$ randomly chooses $U_2$, the second parent, to be any maximal element in its local view (line $2$). 
More specifically, 
\[
    U_1\in\max\{W\in P_A \mid W\geq V\} \text{ and } U_2 \in\max P_A.
\]
Then, the process $A$ chooses a message (it may be empty) and includes it in $U$ (line 3).
Note that we produce new units even if there are no new messages to include.
Then, algorithm \texttt{threshold\_coin} is called, which decides whether a share of the threshold coin should be included in the produced unit and
includes it in such a case (the algorithm \texttt{threshold\_coin} as is explained in detail in Subsection \ref{subsec-leaderless-coin}). 
Finally, the local view of $A$ is updated by adding $U$ to it (line 5).

\begin{algorithm-hbox}[!ht]\caption{\texttt{create\_unit($A$)}}\label{alg_create}
    $U_1\leftarrow$ a random maximal unit in $\P_A$ which is above $A$'s last created unit\\
    $U_2\leftarrow$ a random maximal unit in $\P_A$\\
    Create a unit $U$ with $U_1,\,U_2$ as parents and include a message if available\\
    \texttt{threshold\_coin}($A,U$)\\
    add $U$ to $\P_A$\\
\end{algorithm-hbox}

As described in Algorithm \ref{alg_create} line $1$, each unit created by a process needs to be above all units previously created by this process, i.e. units created by a single process always form a chain in the poset. 
If a faulty process $B$ creates a set $\mathcal F$ of two or more units which share a first parent, then all units from $\mathcal F$ are mutually incomparable and they will be referred to as \emph{forking units}.

The synchronization part is done by means of Algorithm $\texttt{sync}$.
The process $A$ randomly chooses another process $B$ (line $1$) and, if the process $B$ responds, $A$ exchanges information with $B$ about which units each has in its local view (lines $4-5$). 
Using this information, it is possible to exchange only units lacking in $A$'s and $B$'s local views, as done in lines $6-7$.
The received units' signatures are then validated (line $8$) and if no sign of signature forging is detected, the units are added to the local views of $A$.
The update of the local view of a process $A$ by a set of new units $\mathcal{U}_B$ is done one unit at a time in a chosen total ordering extending the ordering of $\mathcal{U}_B$ 
(the choice of this total ordering is irrelevant for the algorithm and does not influence the final total order of messages created by the protocol). 
If at least one signature is not correct, the synchronization is rejected (lines $10-11$) and the connection with $B$ is closed.
The random synchronization scheme of this type is usually referred in the literature as a gossip protocol \cite{Demers}.

\begin{algorithm-hbox}[!ht]
\caption{\texttt{sync($A$)}}\label{alg_sync}
{\bf connect} to a randomly chosen process $B$ and send ID\\
\If {\normalfont received incorrect ID from $B$}
    {
    \Return
    }
{\bf send} information about units in $\P_A$\\
{\bf receive} information about units in $\P_B$\\
{\bf send} set of units $\mathcal{U}_A = P_A\setminus \P_B$\\
{\bf receive} set of units $\mathcal{U}_B = \P_B\setminus \P_A$\\
\If{\normalfont all signatures in $\mathcal{U}_B$ are correct}
    {
        add units in $\U_B$ to the local view $\P_A$\\
    }
\Else
    {
    {\bf reject} synchronization \\
    }
\end{algorithm-hbox}

\subsection{Structural Poset Properties}\label{subsec-structural}

\begin{lemma}\label{lemma-22}
If $U_{0},U_{1}, V_{0}, V_{1}$ are units such that $U_{1} \gg  U_{0}$ and $V_{1} \gg  V_{0} $, then either $U_{1} \geq  V_{0}$ or $V_{1} \geq  U_{0}$. 
\end{lemma}
\begin{proof}
Let $\mathcal U$ be the set of units between $U_{0}$ and $U_{1}$, and let $\mathcal{V}$ be the set of units between $V_{0}$ and $V_{1}$. 
Since $  \#_s\, \mathcal{U} \geq \frac{2}{3}N,\   \#_s\, \mathcal{V} \geq \frac{2}{3}N$
and less than $\frac{1}{3}N$ of the processes are faulty, then there exists a correct process $C\in \supp(\mathcal{U})\cap\supp(\mathcal{V})$.
Because $C \in \mathcal{U} \cap \mathcal{V}$, then $C$ had to create units $U, V$ (not necessarily distinct) such that $U_{1}\geq U\geq U_{0}$ and $V_{1}\geq V\geq V_{0}$. 
Since $C$ is correct, all units created by $C$ are totally ordered, so either $U\geq V$ or $V\geq U$. 
The former implies $U_{1} \geq  U\geq V\geq V_{0}$, the latter $V_{1} \geq  V\geq U \geq U_{0}$.
\end{proof}

\begin{lemma}\label{lemma-two-levels-above}
    If $\P$ is a local view and $U,V\in \P$ are units such that $V\gg U$, then $U$ is high below every unit of level $\L(V)+2$ and higher.
\end{lemma}
\begin{proof}
    Let $W$ be a unit of level $\L(V)+2+l$ for some $l\geq 0$.
    Then, let $\mathcal{U}_{W}$ be the set of all units of level $\L(V)+1+l$ that are high below $W$ and let $\mathcal{U}_{V,U}$ be the set of all units between $V$ and $U$.
    Both of the sets $\mathcal{U}_{W}$ and $\mathcal{U}_{V,U}$ have support of size at least $\frac{2}{3}N$, so there must be a correct process $A\in \supp(\mathcal{U}_{W})\cap\supp(\mathcal{U}_{V,U})$.
    Since $A$ is correct, it had to create units $W_A\in \mathcal{U}_{W}$ and $U_A\in \mathcal{U}_{U,V}$ such that either $U_A\geq W_A$ or $W_A \geq U_A$.
    Since $\L(W_A) = \L(V)+1+l$ and $\L(U_A) \leq \L(V)$, the former is not possible.
    The latter implies that $W\gg W_A \geq U$. 
\end{proof}

\begin{lemma}\label{lemma-growth}
If more than $\frac{2}{3}N$ processes are correct, then the level of the poset will grow indefinitely with probability $1$, even in an asynchronous scenario.
\end{lemma}
\begin{proof}
    Let $\mathcal{P}$ be the global view at some time $t$, and let $\L(\mathcal{P}) = l$.
    We will show that with probability $1$ the poset will grow to level $l+1$. 
    Let $\mathcal{H}$ be the set of all correct processes, and let an \emph{atomic round} be an interval in which at least $\frac{2}{3}N$ processes in $\mathcal{H}$ have issued a unit and synchronized with another process in $\mathcal{H}$. 
    Since correct processes are choosing partners to synchronize randomly, consecutive atomic rounds will continue to occur after time $t$ with probability $1$. 
    
    Let $U$ be a unit in $\mathcal{P}$. 
    Note that information about new units spreads exponentially in the network in terms of atomic rounds.
    Hence, in average, after a logarithmic number of atomic rounds, a unit $V$ will be created that is high above $U$. 
    Let $\mathcal{U}_{U,V}$ be the set of units between $U$ and $V$.
    
    At the time when a unit $W\gg V$ is created, the poset will reach the level $l+1$, since $W$ needs to be high above every $V\in \mathcal{U}_{U,V}$ via transitivity and $\#_s\, \U_{U,V} \geq \frac{2}{3}N$.
\end{proof}

\subsection{Leaderless Threshold Coin}\label{subsec-leaderless-coin}

Usually threshold coin-tossing schemes require a trusted party to secretly deal pairs of private and public keys in the network (see Subsection \ref{coin-description}). 
Since Aleph is designed to work in a leaderless setting, where no such trusted party exists, the protocol requires additional steps to construct a Leaderless Threshold Coin. 

Roughly speaking, the construction will provide access to the common coin for every user at times when the coin is needed during the total ordering procedure.
The core idea is to make every process in the network deal its own threshold coin and then alternate between them, ensuring that at least half of them will be fair. 
We introduce a notation $TC^i(C)$ to denote the value of the threshold coin dealt by $i$-th process for a nonce $C$.
Moreover, we use $TC^i_j(C)$ to denote $j$-th process's share of the threshold coin dealt by $i$-th process for a nonce $C$.

The keys required to construct the threshold coin shares of $TC^i$ will be enclosed by process $i$, the coin's \emph{dealer}, in its first unit produced after the genesis unit which we call the \emph{dealing unit}.

The following notion from the hypergraph theory will be required. 
Let $X$ be a set, and $\mathcal F = \{S_1,\dots, S_k\}$ be a family of subsets such that $S_i\subseteq X$ (i.e., $X$ is a set of vertices and $S$ is a set of hyperedges).
A set $T\subseteq X$ is called a \emph{transversal} of $\mathcal F$ if it has a nonempty intersection with all $S_i$'s. 

When a process $A$ runs Algorithm \ref{algo-threshold-coin}, $A$ first checks whether $U$ is $A$'s first unit after the genesis unit (line $2$). 
In case it is, for every other process $B$, $A$ creates a pair of validation and secret keys for $B$'s share of the coin dealt by $A$ (line $5$), encrypts it with $B$'s public key and embeds the pair within the unit $U$ (line $6$).  
Next, $A$ checks if $U$ is a prime unit of level greater than $4$ (line $7$).
If $U$ satisfies these conditions, $A$ creates the family $\mathcal{F}$ of sets of dealing units below $V$, for every unit $V$ of level less than or equal $\L(U)-3$ which is below $U$ (lines $8-11$). 
Note that here the set of dealing units below each $V$ is a hyperedge and that $\mathcal{F}$ is a set of hyperedges.
Then, a transversal $\mathcal{T}_U$ of $\mathcal{F}$ is constructed iteratively (lines $12-21$),
i.e. new dealing units are added to $\mathcal T_U$ in the order defined by the permutation $\sigma_{\L(U)}$ until they form a transversal. 
The set $\mathcal T_{TC}$ of shares of coins dealt in units forming the transversal $\mathcal T_U$ is then embedded in $U$ (line $22$).

\begin{algorithm-hbox}[!ht]\label{treshold_coin}
\caption{\texttt{threshold\_coin($A, U$)}}\label{algo-threshold-coin}
$i\leftarrow$ $A$'s index\\
\If{{\normalfont$U$ is $A$'s first unit after the genesis unit} }
{
\For{\normalfont$B\in$ processes}
{
$j\leftarrow$ $B$'s index\\
{\bf create} a pair $(VK_j^i,SK_j^i)$ of keys for $B$'s share of coin $TC^i_j$\\
{\bf embed} $(VK_j^i, SK_j^i)$ encrypted with $B$'s public key $PK_B$ within $U$
}
}
\If{{\normalfont$U \in \P_A$ is a prime unit and $\L(U)\geq 4$}}
    {
    $\mathcal F \leftarrow \emptyset$\\
    \For{\normalfont $V\in \big\{W\in\P_A : W \leq U,\, \L(W) \leq \L(U) - 3,\, W\in \pri(P)\big\}$}
        {
        $S\leftarrow$ set of all dealing units below $V$\\
        {\bf append $S$ to $\mathcal F$} 
        }
    $\mathcal T_U \leftarrow \emptyset $\\
    $\mathcal T_{TC} \leftarrow \emptyset$\\
       $k \leftarrow 1$\\
    \While{$\mathcal T_U$ \normalfont is not a transversal of $\mathcal F$}
        {
        $j\leftarrow \sigma_{\L(U)}(k)$\\
        $B\leftarrow j$-th process\\
        \If{\normalfont$B$'s dealing unit is in $\cup\mathcal F$}
            {
            {\bf append} $B$'s dealing unit to $\mathcal T_U$\\
            {\bf embed} $TC_{i}^{j}(\L(U))$ within $\mathcal T_{TC}$
            }
        $k \leftarrow k+1$ \\
        }
    {\bf append} $\mathcal T_{TC}$ to $U$
}
\end{algorithm-hbox}

While it may look like the protocol requires a lot more overhead due to the necessity of storing coin shares in units, 
the following lemma shows that only a very small number of such shares are necessary on average.

\begin{lemma}\label{trans-size}
    The expected number of units in the transversal $\mathcal T_U$ constructed in Algorithm \ref{algo-threshold-coin} is not bigger than $\frac{3}{2}$.
\end{lemma}

\begin{proof}
    Let $A$ be a process, $U$ be a prime unit created by $A$ of level at least $4$, and $\mathcal{F}$ be the family created as in Algorithm \ref{algo-threshold-coin} for which the transversal $\mathcal T_U$ is constructed. 
    Moreover, let $V\in \P_A$ be a prime unit such that $U\geq V$ and $1\leq \L(V)\leq \L(U)-3$.
    Due to Lemma \ref{lemma-two-levels-above}, every unit of level $\L(U)$ is high above every unit high below $V$.
    Hence, the set $\down{V}$ of at least $\frac{2}{3}N$ dealing units is in every set in $\mathcal{F}$.
    Thus, in any permutation $\sigma$, every dealing unit has a probability at least $\frac{2}{3}$ of being the transversal by itself.
    By summing the geometric sequence, we obtain that the transversal is constructed after at most $\frac{3}{2}$ elements on average.
\end{proof}

Due to the way coin shares are included in the units, not every coin value will be available on every level. 
To control which coins will be available, we introduce the following notion.

\begin{definition}
    Let $k\geq 1$ and $V$ be a prime unit.
    We define a \emph{first available index for level $k$ and the unit $V$} ($\fai_k(V)$ for short)
    to be the lowest index $i$ in $\sigma_k$ such that $i$-th process's dealing unit is below $V$.
\end{definition}

\begin{lemma}\label{lemma-coin-share}
    Let $A$ be a process, $\P_A$ be a local view, let $V\in \P_A$ be a prime unit, let $U\in \P_A,\ U\geq V$ be a prime unit of level at least $\L(V)+3$, created by the $i$-th process.
    Then, $U$ contains the coin share $TC_i^{\fai_{\L(U)}(V)}(\L(U))$.
\end{lemma}

\begin{proof}
    In Algorithm \ref{algo-threshold-coin}, the procedure of adding coin shares to the unit $U$ continues until a transversal of $\mathcal{F}$ is constructed.
    Since the set of all dealing units below $V$ belongs to $\mathcal{F}$ and coin shares are added to $U$ in the order defined by $\sigma_{\L(U)}$,
    we have that the $\fai_{\L(U)}(V)$-th share of the coin of $i$-th process for the nonce $\L(U)$, or in our notation $TC_i^{\fai_{\L(U)}(V)}(\L(U))$, has to be contained in the unit $U$.
\end{proof}

\begin{lemma}\label{lemma-coin-compute}
    Let $A$ be a process and let $V\in \P_A$ be a prime unit.
    Then, for every $k$ such that $\L(\P_A) > k \geq \L(V)+3$, the value of the coin $TC^{\fai_k(V)}(k)$ is possible to compute in $\P_A$.
\end{lemma}

\begin{proof}
    Since $\L(\P_A) > k$, $\P_A$ contains a set $S$ of at least $\frac{2}{3}N$ prime units at level $k$ produced by different processes.
    Fix $U\in S$ and let $W$ be a prime unit of level $\L(V)+1$ created by $A$.
    By applying Lemma \ref{lemma-two-levels-above} to $W$ and $V$, we obtain that $U \geq V$, as $\L(U)\geq \L(W)+2$.

    Hence, we obtain by Lemma \ref{lemma-coin-share} that, every unit in $S$ contains its share of the coin $TC^{\fai_k(V)}(k)$,
    which allows us to reach the required threshold and compute the coin value.
\end{proof}


\subsection{Timing units}\label{subsec-timing}

The algorithm for establishing the total ordering of all units has two stages. 
First, it chooses some prime units to be \emph{timing units} (Definition \ref{timing-unit}), which serve as common reference points for the total ordering.
Then, all the remaining units are totally ordered with respect to the timing units.
The main challenge is to guarantee a consistent choice of the timing units.
We want to ensure that different processes choose the same timing units despite differences in their local views.

In this subsection, we define auxiliary notions needed for choosing timing units and describe their fundamental properties.
The idea for the functions $\Pi$ and $\Delta$ as well as the use of the threshold coin for breaking stalemates was inspired by the ABBA protocol introduced in \cite{Cachin}.
The algorithm for the total ordering is presented and analyzed in the next subsection.

On every level $l$ of a poset, we choose one timing unit among all prime units on that level.
The choice is done in two phases:
\begin{enumeratenum}
\item in the \emph{decision phase}, for every prime unit of level $l$, the decision on its feasibility as a timing unit is made by means of the function $\Delta$, and
\item in the \emph{choice phase}, among all feasible prime units of level $l$, the one which is first in the permutation $\sigma_l$ is chosen to be the timing unit for that level.
\end{enumeratenum}

The decision phase is performed for every prime unit separately.
Throughout this section $U_c$ denotes a fixed prime unit which is considered as a candidate for a timing unit.
The decision phase for $U_c$ is organized into rounds, during which other prime units decide on the feasibility of $U_c$ as a timing unit.
Each round consists of computing values of functions $\Delta_{U_c}$ and $\Pi_{U_c}$ for prime units belonging to one level of the poset.
During $k$-th round ($k=0,1,\dots$) all prime units with level $\L(U_c)+k$ are considered.
There are three types of rounds: preliminary rounds (denoted with $\bot$), even rounds (denoted with $0$), and odd rounds (denoted with $1$).

Before stating the definitions of $\Delta_{U_c}$ and $\Pi_{U_c}$, we need to introduce several auxiliary functions.
The first one computes the type of round to which a prime unit $U$ belongs:
\[
    \pa(U) = 
      \begin{cases*}
      (\L(U)-\L(U_c)) \bmod{2} & if $\L(U) > \L(U_c)+1$\\
      \bot & if $\L(U) \leq \L(U_c)+1$,
    \end{cases*}
\]

Next, we define operators counting results of any given function $f$ at the previous level. 
The supermajority operator $\SupMaj(S,f)$ counts values of function $f$ in the set $S$ and returns the one with at least $\frac{2}{3}N$ occurrences, or $\bot$ in case no such value exists:
\[
    \SupMaj(S,f) =
    \begin{cases*}
      value & if $\#_s\{U \in S: f(U) = value\} \geq \frac{2}{3}N$\\
      \bot & otherwise
    \end{cases*}
\]

The existence operator $\Exist(S,f)$ returns $0$ or $1$ if such value exists in $f(S)$ (with the preference for $1$), or $\bot$ otherwise:
\[
    \Exist(S,f) =
    \begin{cases*}
      1 & if $1 \in f(S)$\\
      0 & if $1 \not\in f(S)$ and $0\in f(S)$\\
      \bot & otherwise.
    \end{cases*}
\]
The randomized existence operator $\Exist_{TC}(S,f)$ returns $0$ or $1$ if such a value exists in $f(S)$ (with the preference for $1$),
otherwise it returns the value of the threshold coin $\TC^{\fai_{\L(S)}(U_c)}(\L(S))$, where $\L(S) = \max \{\L(U) : U \in S \}$.
Let us note that we will use only sets $S$ of constant level, i.e., $\L(S)=\L(U)$ for all $U\in S$.
\[
    \Exist_{TC}(S,f) =
    \begin{cases*}
      1 & if $1 \in f(S)$\\
      0 & if $1 \not\in f(S)$ and $0\in f(S)$\\
      \TC^{\mathrm{\fai}_{\L(S)}(U_c)}(\L(S)) & otherwise.
    \end{cases*}
\]

Using the above operators, we recursively define the proposal function $\Pi_{U_c}$ and the decision function $\Delta_{U_c}$:

\[
    \Pi_{U_c}(U) =
    \begin{cases*}
      \SupMaj(\down{U} ,\Pi_{U_c}) & if $\pa(U) = 1$\\
      \Exist_{TC}(\down{U} ,\Pi_{U_c}) & if $\pa(U) = 0$\\
      [U \geq U_c] & if $\pa(U) = \bot$
    \end{cases*}
\]
and
\[
\Delta_{U_c}(U) =
    \begin{cases*}
       \Exist(\down{U} ,\Delta_{U_c}) & if $\pa(U) = 1$ \\
       \SupMaj(\down{U},\Pi_{U_c}) & if $\pa(U) = 0$ \\
       \bot & if $\pa(U) = \bot$.
    \end{cases*}
\]

The core of the decision phase is the $\Delta_{U_c}$ function.
As proven in the following technical lemma, it provides the decision procedure with the desired properties of unanimity, finiteness and necessity of positive decisions.  

\begin{lemma}\label{lemma-technical-ordering}
    Let $A$ be a process, let $\P_A$ be a local view, let $U_c\in \P_A$ be a prime unit.
    The functions $\Pi_{U_c}$ and $\Delta_{U_c}$ safisfy the following properties:
    \begin{enumeratei}
        \item if $\Pi_{U_c}(U)\neq\bot$ for some prime unit $U\in \P_A$ such that $\pa(U) = 1$,
            then $\Pi_{U_c}(V) \in \{\Pi_{U_c}(U), \bot\}$ for every prime unit $V\in \P_A$ such that $\L(U) = \L(V)$ {\bf (single supermajority)},
        \item if $\Delta_{U_c}(U)\neq\bot$ for a prime unit $U\in \P_A$,
            then $\Delta_{U_c}(V) = \Delta_{U_c}(U)$ for every prime unit $V\in \P_A$ such that $\L(V) \geq \L(U)+2$ {\bf (unanimous decision)}, 
        \item if $\Pi_{U_c}(U) = TC^{\fai_{l-1}(U_c)}(l-1)$ for every prime unit $U\in \P_A$ of level $l = \L(U_c)+2k$, for some fixed $k>1$,
            then for every prime unit $V$ of level $l+2$ or higher there is $\Delta_{U_c}(V)\neq\bot$ {\bf (finiteness of decision procedure)},
        \item if $U_1\in \P_A$ is a prime unit of level $\L(U_1)=l+1$, $l \geq 0$, such that there is no other prime unit of that level below $U_1$,
            then for every $U_c\in \down{U_{1}}$ and every prime unit $U_4\in \P_A$ of level $l+4$ there is $\Delta_{U_c}(U_4) = 1$ {\bf (necessity of positive decisions)}, and
		\item if $V\in \P_A$ is a prime unit such that $\L(V) = \L(U_c) + 4$ and $U_c\nless V$,
            then $\Delta_{U_c}(W)\in \{\bot, 0\}$ for every $W$ of level at least $\L(U_c)+4$ {\bf (negative decisions for hidden units)}.
    \end{enumeratei}
\end{lemma}
\begin{proof}

$ $\newline
${\bf (i)}$ Suppose by way of contradiction that for some prime unit $V\in\P_A$ with $\L(U)=\L(V)$, there is $\Pi_{U_c}(U) = 1$ and $\Pi_{U_c}(V) = 0$. 
By the definition of $\Pi_{U_c}$, we have that $\SupMaj(\down{U} ,\Pi_{U_c}) = 1$ and $\SupMaj(\down{V},\Pi_{U_c}) = 0$.  This implies that 
on level $\L(U)-1$ there are sets $S_U$ and $S_V$ containing at least $\frac{2}{3}N$ prime units issued by different processes, such that $\Pi_{U_c}(U') = 0$ and $\Pi_{U_c}(V') = 1$ for all $U'\in S_U,\ V'\in S_V$.
Since both $S_U$ and $S_V$ are composed of units issued by at least $\frac{2}{3}N$ different processes, then there exists a correct process
who has issued a prime unit that is in both $S_U$ and $S_V$.
This gives a contradiction, since a correct process could issue at most one prime unit of the given level,
which could belong to only one of the sets $S_U, S_V$, depending on its $\Pi_{U_c}$ value.

\medskip

\noindent ${\bf (ii)}$ Let $U$ be a prime unit of the lowest possible level such that $\Delta_{U_c}(U)\neq \bot$.
One can readily observe that the general case follows from this one.
For an odd round, the function $\Delta$ differs from $\bot$ only in the case when $\Delta$ is different from $\bot$ on the previous level, hence we have $\pa(U) = 0$.
The definition of $\Delta$ implies that there is a set $S_U\subset \down{U_c}$ of at least $\frac{2}{3}N$ prime units of level $\L(U)-1$ created by different processes and that $\Pi_{U_c}(V) = \Delta_{U_c}(U)$ for every $V\in S_U$. 
Because $\#_sS_U>\frac{2}{3}N$, at least $\frac{1}{3}N+1$ units were created by correct processes. 
Then, every prime unit $W$ of level $\L(U)$ needs to be high above at least one unit in $S_U$ issued by a correct process.
Hence,
\[
    \Pi_{U_c}(W) = \Exist_{TC}(\down{U} ,\Pi_{U_c}) = \Delta_{U_c}(U).
\]
As a result, for every unit $W_1$ of level $\L(U)+1$,
\[
    \Pi_{U_c}(W_1) = \SupMaj(\down{W_1} ,\Pi_{U_c}) = \Delta_{U_c}(U).
\]
Finally, for every unit $W_2$ of level $\L(U)+2$, we obtain
\[
    \Delta_{U_c}(W_2) = \SupMaj(\down{W_2},\Pi_{U_c}) = \Delta_{U_c}(U).
\]

\medskip

\noindent ${\bf (iii)}$ 
By the assumption, we have for some $v\in \{0,1\}$ that $\Pi_{U_c}(U)=v$ for every prime unit $U$ of level $l$.
By the definition of $\Pi_{U_c}$, for every $W$ of level $l+1$ and above, $\Pi_{U_c}(V) = v$.
Finally, by the definition of $\Delta$, for any $V$ of level at least $l+2$, $\Delta_{U_c}(V)\neq \bot$. 

\medskip

\noindent ${\bf (iv)}$ Let $U_c$ be a unit in $\down{U_{1}}$ and $S_{U_c,U_1}$ be the set of all units between $U_c$ and $U_1$.
Since $\#_s\, S_{U_c,U_1} > \frac{2}{3}N$, there exists at least $\frac{1}{3}N+1$ correct processes in $\supp(S_{U_c,U_1})$. 
Now let $U_2$ be any prime unit of level $l+2$ and $A$ be a correct process in $\supp(S_{U_c,U_1})\cap \supp(\down{U_2})$.
Moreover, let $U_A$ be a unit issued by $A$ in $S_{U_c,U_1}$ and $U'_A$ be the prime unit issued by $A$ on level $l+1$. 
Since $A$ is correct, $U_A\leq U_1$ and $U'_A\nless U_1$ (due to the assumption that no prime unit of $U_1$'s level is below it), we get that $U'_A\geq U_A$ and hence $U_2\gg U'_A\geq U_A\geq U_c$.
Since $U'_A\geq U_c$, we have that
\[
	\Pi_{U_c}(U'_A) = [U'_A \geq U_c] = 1.
\]

\noindent Next, due to $U_2\gg U'_A$, we also have that
\[
	\Pi_{U_c}(U_2) = \Exist_{TC}(\down{U_2} ,\Pi_{U_c}) = 1.
\]

\noindent Since the choice of $U_2$ was arbitrary, we observe that $\Pi(U_2', U_0)=1$ for every prime unit $U_2'$ of level $l+2$.
Hence 
\[
    \Pi_{U_c}(U_3) = \SupMaj(\down{U_3},\Pi_{U_c}) = 1
\]
\noindent for every $U_3$ of level $l+3$.
Finally, for every $U_4$ of level $l+4$ we obtain 
\[
    \Delta_{U_c}(U_4) = \SupMaj(\down{U_4},\Pi_{U_c}) = 1.
\]

\medskip

\noindent ${\bf (v)}$ Let $U_c$ and $V$ be as in the assumptions.
By the contraposition of Lemma \ref{lemma-two-levels-above}, we have that no prime unit of level $\L(U_c)+2$ can be high above $U_c$. 
Hence, for every $U_1 < U_2 < U_3 < U_4$ of levels $\L(U_i) = \L(U_c)+i$, the following holds:
\[
\Pi_{U_c}(U_1) = [U_1 \geq U_c] = 0
\]
\[
\Pi_{U_c}(U_2) = \Exist_{TC}(\down{U_2} ,\Pi_{U_c}) = 0,
\]
\[
\Pi_{U_c}(U_3) = \SupMaj(\down{U_3} ,\Pi_{U_c}) = 0, and
\]
\[
\Delta_{U_c}(U_4) = \SupMaj(\down{U},\Pi_{U_c}) = 0.
\]

\noindent Since $U_4$ was chosen arbitrarily, the same is true for any unit $W$ of level $\L(U_c)+4$ and for all levels above.
This is due to the recursive definition of $\Delta$ for odd rounds and $(ii)$.  

\end{proof}


The following two definitions formalize the concepts of the decision and choice phases that are required for the selection of a timing unit. 
 
\begin{definition}
    Let $A$ be a process.
    Let $\P_A$ be a local view and $U\in \P_A$ be a prime unit.
    We say that $\P_A$ \emph{decides $b\in\{0,1\}$ on $U$} if there exists a level $l>\L(U)$ such that for all prime units $V$ of level $l$, we have that $\Delta_U(V)=b$.
\end{definition}

\begin{definition}
    Let $A$ be a process.
    Let $\P_A$ be a local view and $U\in \P_A$ be a prime unit  created by a process $B$ of index $i$.
    We say that $\P_A$ \emph{chooses $U$} if the following conditions are met:
    \begin{enumeratenum}
    	\item $\P_A$ decides $1$ for $U$, and
        \item for all $j$ prior to $i$ in the ordering given by $\sigma_{\L(U)}$, a prime unit or a set of forking prime units created by a process $j$ on the level $\L(U)$ are decided $0$ by $\P_A$.
    \end{enumeratenum}
\end{definition}

\begin{remark}
    Let $A$ be a process and $t$ be a time when the local view $\P^t_A$ chooses a prime unit $U$.
    There is a possibility that at some later time $t^\prime > t$, a prime unit $U^\prime$ of level $\L(U)$ that is prior to $U$ in the ordering given by $\sigma_{\L(U)}$
    will be added to the local view $\P^{t^\prime}_A$.
    By $(v)$ of Lemma \ref{lemma-technical-ordering}, such a unit can only be decided with a value of $0$ in $\P^{t^\prime}_A$.  Hence, the above notion of choosing is well defined.
\end{remark}

Using the above construction, we are now able to define timing units, which have a crucial property of being unanimously chosen by all correct processes. 
This property makes them suitable for setting a common chronology of events in an asynchronous scenario.

\begin{definition}\label{timing-unit}
    Let $A$ be a process.
    Let $\P_A$ be a local view.
    A chain of units $T_0, T_1,... \in \P_A$ of levels $\L(T_i) = i$, which are chosen by $\P_A$, is called \emph{a timing chain}, and each $T_i$ is called \emph{a timing unit}.
\end{definition}

In the next two results, we provide bounds on the average number of levels after which a decision on a prime unit occurs (Lemma \ref{lemma-decision})
and the average number of levels after which a timing unit is chosen (Theorem \ref{lemma-choice}). 

\begin{lemma}\label{lemma-decision}
    Let $A$ be a process.
    At the moment of a creation of a prime unit $U_c$ in a local view $\P_A$,
    the expected number of rounds required for the decision on $U_c$ in $\P_A$ is not bigger than $10$. 
\end{lemma}
\begin{proof}
    During the process of deciding for $U_c$, coin values will be chosen among all threshold coins dealt by units high below $U_c$, i.e., the coins dealt by at least $\frac{2}{3}N$ different processes. 
    Since there is less than $\frac{1}{3}N$ faulty processes, over half of these coins are dealt by a correct process and hence their result cannot be predicted until enough units with their shares are produced.
    Within the scope of this proof, we will refer to such coins as a \emph{fair}. 
    The above implies that there is at least $\frac{1}{2}$ probability that the coin $TC^{\fai_l(U_c)}(l)$ is fair, for every $l\geq \L(U_c)+3$.
    
    Fix $k>0$, and let $U$ be a prime unit of level $l = \L(U_c) + 2k$.
    The coin used to compute $\Pi_{U_c}(U)$ is $TC^{\fai_{l-1}(U_c)}(l-1)$,
    and when the poset reaches level $l$ its value is possible to compute by Lemma \ref{lemma-coin-compute}.
    
     Now, let $v\in \{0,1\}$ be a value (which exists by Lemma \ref{lemma-technical-ordering}$(i)$) such that for every $V$ of level $l-1$, there is $\Pi_{U_c}(V) \in \{v, \bot\}$. 
     Then, for every $U_l$ of level $l$ we have
    \[
        \Pi_{U_c}(U_l) = \Exist_{TC}(\down{U_l} ,\Pi_{U_c}) \in \{v, TC^{\fai_{l-1}(U_c)}(l-1)\} .
    \]
    After the coin reveal, if $TC^{\fai_{l-1}(U_c)}(l-1) = v$, then, in the next round for every $U_{l+1}\in S_{l+1}$, 
    \[
        \Pi_{U_c}(U_{l+1}) = \SupMaj(\down{U_{l+1}} ,\Pi_{U_c}) = v.
    \]
    As a consequence, in the next round, we obtain
    \[
        \Delta_{U_c}(U_{l+2}) = \SupMaj(\down{U_{l+2}},\Pi_{U_c}) = v
    \]
    for every $U_{l+2}\in \P_A$ of level $l+2$. 
    Finally, applying Lemma \ref{lemma-technical-ordering}$(ii)$ to $\Delta_{U_c}(U_{l+2})$,
    we observe that if $TC^{\fai_{l-1}(U_c)}(l-1) = v$, then the unit $U_c$ is decided on the level $l+4$.
    
    Note that if $TC^{\fai_{l-1}(U_c)}$ is a fair coin then $P\big(TC^{\fai_{l-1}(U_c)} = v\big) = \frac{1}{2}$, and hence the unconditional probability that $TC^{\fai_{l-1}(U_c)}(l-1) = v$ is at least $\frac{1}{4}$.     
    From the above considerations we have that the expected number of rounds necessary to decide on $U_c$ is not greater than
    \[
         4+\nicefrac{1}{4}\sum_{i\geq 0}2i\big(\nicefrac{3}{4}\big)^i = 4+\frac{1}{2}\sum_{i\geq 0}i\big(\frac{3}{4}\big)^i= 4+\frac{1}{2}\cdot 12=10,
    \]
    where the third equality comes from the summation of the series via Lemma \ref{lemma-spanish}. 
\end{proof}

\begin{theorem}\label{lemma-choice}
    At the moment of a creation of the first prime unit of level $l+2$, $l \geq 0$, the expected number of rounds required for the choice of a timing unit on the level $l$ is less than $9$.
\end{theorem}
\begin{proof}
    Let $U_1$ be a prime unit of the level $l+1$, such that no other prime unit of that level is below it. 
    By Lemma \ref{lemma-technical-ordering}(iv), we know that every unit $U_c\in \down{U_1}$ will be positively decided at the level $l+4$, and since $\#_s(\down{U_1})\geq \frac{2}{3}N$,
    we have that the probability that a timing unit will be chosen after $4$ rounds is greater or equal to $\frac{2}{3}$. 
    On the other hand, by Lemma \ref{lemma-decision}, we know that the number of rounds to decide on other units is at most $10$ on average.
    
    By bounding the expected number of rounds required to decide on $i$ units from outside of $\down{U}$ by $10i$,
    we obtain that the expected number of rounds needed for the choice of a timing unit is not greater than
    \[ 
    \sum_{i \geq 0}\big(\frac{1}{3}  \big)^{i}\big(\frac{2}{3}\cdot 4+\frac{1}{3}\cdot 10 \big) = 6\cdot \sum_{i \geq 0}\big(\frac{1}{3}  \big)^{i} = 9
    \]

\end{proof}
\subsection{Total Ordering of Units}\label{subsec-ordering}

Equipped with our well-defined and consistent notion of timing units we are ready to present the algorithm for establishing a total order of units.
We prove that our algorithm satisfies all essential requirements, i.e., correctness, integrity, and liveness.

For a unit $U$, its \emph{timing round} is the minimal $k$ such that $T_k \geq  U$, where $T_k$ is the $k$-th timing unit. 
Since many units can have the same timing round, the ties will be settled via Algorithm $\texttt{break\_ties}$. 
At first, the algorithm computes hashes of all the units in $\mathcal{U}$ with a hashing function $\phi$ (line $1$) and then computes $R$ as the XOR of all of them (line $2$).
The XOR operation is chosen as it is not possible to influence the value of $R$ without controlling every hash and due to its speed.  Next, for each unit $U$, a new attribute \texttt{tiebreaker} is computed as the XOR of $R$ and $U$'s hash (lines $3-4$).
Then, the set $\mathcal{U}$ is topologically sorted settling ties via $\texttt{tiebreaker}$ values (lines $6-10$), and the total ordering is returned (line $11$).

\begin{algorithm-hbox}[!ht]
\caption{\texttt{break\_ties($A, \mathcal{U}$)}}\label{break_ties}
\tcp{$ \phi$ - hashing function}
$\mathcal{H} \leftarrow \emptyset$

\For{$U \in \mathcal{U}$}
{
$\mathcal{H} \leftarrow \mathcal{H} \cup \{\phi(U)\}$
}

$R \leftarrow \mathrm{XOR}_{H_i \in \mathcal{H}} H_i$\\
\For{$U \in\mathcal{U}$}
{
$\texttt{U.tiebreaker}\leftarrow \phi(U)\ \mathrm{XOR}\ R$    
}

$L\leftarrow $ empty ordered list\\
\While{$\mathcal{U}\neq\emptyset$}
{
$M\leftarrow$ minimal elements of $\mathcal{U}$\\
$\mathcal{U}\leftarrow\mathcal{U}\setminus M$\\
{\bf sort} $M$ by tiebreaker values\\
{\bf extend} $L$ with $M$\\
}
\Return $L$
\end{algorithm-hbox}

Now we are ready to state the final theorem, proving the most fundamental properties of the Aleph protocol:

\begin{theorem}\label{thm-ordering}
The Aleph protocol for agreement on total ordering of units satisfies the following properties:
\begin{enumerate}
\item[-] {\bf Correctness} - all correct processes will end up with the same ordering,
\item[-] {\bf Integrity} - All decisions made by correct processes are irreversible,
\item[-] {\bf Liveness} - all correct processes will eventually decide with probability $1$.
\end{enumerate}
\end{theorem}

\begin{proof} $ $\\
\noindent {\bf Correctness}
Due to Lemma \ref{lemma-technical-ordering}(ii), correct processes decide unanimously on whether a prime unit is chosen to be a timing unit.
Since the ordering is deterministic after timing units are agreed upon, all correct processes are bound to create the same total order of units. 

\noindent {\bf Integrity} Again by Lemma \ref{lemma-technical-ordering}(ii), we observe that when the function $\Delta$ decides on a unit $U$ at some level $l$, the decision is bound to stay the same on all levels $>l$. 
Since the rest of the protocol is based solely on these fundamental decisions, integrity naturally follows. 

\noindent {\bf Liveness} 
Since new levels in the poset are guaranteed to be created by Lemma \ref{lemma-growth}, we have that the poset will reach any required level eventually. 
Let $U$ be a unit in the poset. If $U$ is issued by a correct process, then $U$ will eventually be gossiped to at least $\frac{2}{3}N$ other correct processes, which will eventually build their units above it. 
After correct processes issue units above $U$ in level $l$, all timing units issued by the correct processes at level $l+1$ will be above $U$.  There exists at least one of these units, say $V$, such that  there is a prime unit of level $l+1$ high above $V$.  As such, from Lemma \ref{lemma-two-levels-above}, we obtain that all prime units of level $l+4$ and above are high above $U$. 
Finally, all timing units chosen among prime units of levels $l+4$ and above will be high above $U$; hence, $U$ will be ordered as soon as any timing unit is chosen and occurs on level $l+4+6$ on average, via Theorem \ref{lemma-choice}.

\end{proof}

\subsection{Fast Unit Validation}\label{subsec-validation}

Consensus on a total ordering of all messages is a complex task and requires multiple rounds of communication between all processes. 
In this section we prove a much faster mechanism, guaranteeing that only one version of each message will be accepted by the network, without the need of ordering of all messages. 
One of its potential use cases is in a token exchange system, where this scheme can provide a method for quick digital asset transfers without the risk of two or more conflicting transactions being accepted. 
The mechanism is based on the following definition:

\begin{definition}
    A unit $U$ will be considered \emph{validated} for a process $A$ if there is a unit $V$ in its local view such that $V$ is high-above $U$ and $V$ is not above any forks of $U$ (even if some exists in $A$'s local view). 
    We will denote this in short as ${\normalfont\val_A(U) = \texttt{True}}$. 
    In such a scenario, the unit $V$  will be referred to as $U$'s \emph{validator}.
\end{definition}

Note that a given unit can have more than one validator.     
A straightforward application of Lemma \ref{lemma-22} gives the following corollary.

\begin{corollary}\label{lemma-onechecked}
    No two forking units can both be validated in the poset.
\end{corollary}
\begin{proof}
    On the contrary, assume that $U_{0}$ and $U_{1}$ are forking units that are both validated by $V_{0},V_{1}$ respectively.
    Since $V_{0}\gg U_{0} $ and $V_{1}\gg U_{1}$, we obtain by Lemma \ref{lemma-22}, that either $V_{0}\geq U_{1}$ or $V_{1}\geq U_{0}$, contradicting validation of one of the units. 
\end{proof}

The following theorem states that the validation property of a unit satisfies the fundamental requirements of Byzantine agreement.

\begin{theorem}\label{thm-validation}
    Validation of units satisfies the following properties while tolerating up to $\frac{1}{3}$ faulty processes: 
    \begin{enumeratei}
        \item[-] {\bf Correctness} - no two correct processes can decide differently on validity of a unit,
        \item[-] {\bf Integrity} - at most one version of every unit can get validated by a correct process,
        \item[-] {\bf Liveness} - every correct process eventually validates each non-forking unit with probability $1$.
    \end{enumeratei}
\end{theorem}
\begin{proof}$  $\par\nobreak\ignorespaces

    {\bf Correctness.} If a correct process $A$ decides that a unit $U$ is valid, then there exists a unit $V \in \P_A$ such that $V$ validates $U$.
    While $A$ continues gossiping, all other correct processes will eventually receive the unit $V$, and hence validate $U$.
    
    \medskip
    
    {\bf Integrity.} This a direct consequence of Corollary \ref{lemma-onechecked}. 
    
    \medskip
    
    {\bf Liveness}. Let $U$ be a non-forking unit created by a process $A$.
    Due to Lemma \ref{lemma-growth}, the unit $U$ will eventually reach $\frac{2}{3}N$ correct processes.
    Due to the random choice of the second parent of every unit, each of these correct processes will inevitably issue a unit above $U$ and gossip it back to $A$, either directly, or via other correct processes.
    The process $A$, after adding $\frac{2}{3}N$ units above $U$ created by different processes, is bound to eventually create a unit $V$ above all of them, and hence validate the unit $U$.
    The unit $V$, will then reach all the other correct processes, again, due to Lemma \ref{lemma-growth}.
\end{proof}

Note that this concept is not meant as a stand-alone mechanism, but rather as an additional feature to be used with the main ordering algorithm. 
Let us consider a case where messages are token transactions and all of the processes shares the ledger with all accounts. 
If a user $A$ has funds on his account and is willing to transfer it, ordering this transaction according to all other transaction is not crucial. 
What is crucial, however, is to ensure that this transaction will be accounted for after all previous $A$'s transactions and that there will be no other version of this transaction accepted by the network (i.e., no conflicting transaction).
This can be achieved utilizing the above definition of validation -- if some $A$'s transaction $T$ is conveyed in a validated unit, and all previous $A$'s transactions are validated, then $T$ can be safely accounted for, provided that $A$ has enough funds on the account.
If however, $A$ has insufficient funds, the system needs to wait for the total ordering, since there is a possibility that there was a transaction issued by some other node then $A$ transferring funds to $A$'s account.

\section{Acknowledgements}
The authors would like to thank Johan Bratt for introducing us to the field, constant support, and many hours spent on valuable discussions.
We would also like to show our gratitude to Matthew Niemerg for reading our paper countless times, proposing changes that greatly improved consistency and readability of the paper, and discussions that helped us clarify notions and proofs contained in the paper.
We are also immensely grateful to Michał Handzlik for his invaluable help during the last phase of writing the manuscript.

This research was funded by the Aleph Zero Foundation.


\begin{thebibliography}{99}

\bibitem{Baird}  Baird L. \emph{The Swirlds Hashgraph Consensus Algorithm: Fair, Fast, Byzantine Fault Tolerance}, Swirlds Tech Report SWIRLDS-TR-2016-01 (2016)

\bibitem{Cachin} Cachin, C., Kursawe, K., Shoup, V. \emph{Random Oracles in Constantinople: Practical Asynchronous Byzantine Agreement using Cryptography}, Journal of Cryptology 18(3), 2000

\bibitem{CVNV}  Correia, M., Veronese, G. S., Neves, N. F., Verissimo, P. \emph{Byzantine consensus in asynchronous message-passing systems: a survey}, Int. J. Critical Computer-Based Systems, Vol. 2, No. 2, 2011

\bibitem{Demers} Demers, A., Greene, D., Hauser, C., Irish, W., Larson, J., Shenker, S., Sturgis, H., Swinehart, D., Terry. \emph{Epidemic Algorithms for Replicated Database Maintenance}, Proceedings of the Sixth Annual ACM Symposium on Principles of Distributed Computing. PODC '87.

\bibitem{FLP} Fisher, M., Lynch, N., Paterson, M. \emph{Impossibility of distributed consensus with one faulty process}, Journal of ACM, Volume 32 Issue 2, 1985, 374-382.

\bibitem{glss} Gagol, A., Lesniak, D., Straszak, D., Swietek, M. \emph{Aleph: Efficient Atomic Broadcast in Asynchronous Networks with Byzantine Nodes}, https://arxiv.org/abs/1908.05156

\bibitem{Lamport} Lamport, L.  \emph{The Part-Time Parliament}, ACM Transactions on Computer Systems. 16 (2): 133–169, 1998.

\bibitem{Lamport-clocks} Lamport, L.  \emph{Time, clocks, and the ordering of events in a distributed system}, Communications of the ACM 21 (7): 558-565, 1978.

\bibitem{OO} Ongaro, D., Ousterhout, J. \emph{In Search of an Understandable Consensus Algorithm} 

\bibitem{Pass} Pass, R., Seeman, L., Shelat, A. \emph{Analysis of the blockchain protocol in asynchronous networks}, Eurocrypt, 2017.

\bibitem{Sat} Satoshi, N. \emph{Bitcoin: A Peer-to-Peer Electronic Cash System}, https://bitcoin.org/bitcoin.pdf

\end{thebibliography}
\end{document}